\documentclass[10pt]{IEEEtran}
\usepackage[utf8]{inputenc}
\usepackage{amsmath}
\usepackage{amsthm}
\usepackage{amssymb}
\usepackage{graphicx}
\usepackage{psfrag}
\usepackage{cite}
\usepackage{url}
\usepackage{enumerate}
\usepackage{color}

 \theoremstyle{plain}% default
 \newtheorem{thm}{Theorem}%[section]
 \newtheorem{lem}{Lemma}
 
 \newtheorem{cor}{Corollary}
 \theoremstyle{definition}
 \newtheorem{defn}{Definition}%[section]
 \newtheorem{rem}{Remark}

 \newtheorem{assm}{Assumption}

\newcommand{\bb}[0]{\begin{bmatrix}}
\newcommand{\eb}[0]{\end{bmatrix}}
\newcommand{\be}[0]{\begin{equation}}
\newcommand{\ee}[0]{\end{equation}}
\newcommand{\ben}[0]{\begin{equation*}}
\newcommand{\een}[0]{\end{equation*}}

\newcommand{\norms}[1]{\lVert#1\rVert}

\let\norm=\normB

\renewcommand{\Re}[0]{\mathbb R}
\newcommand{\Ce}[0]{\mathbb C}

\usepackage{enumitem}

\title{Adaptive Output Feedback based on Closed-loop Reference Models}
\author{Travis~E.~Gibson, Zheng Qu, Anuradha M. Annaswamy and Eugene Lavretsky
\thanks{T.~E. Gibson is with  Harvard Medical School and the Channing Division of Network Medicine, Department of Medicine, Brigham and Women's Hospital, Boston MA 02115 email:({travis.gibson@channing.harvard.edu}). Previously with the Department of Mechanical Engineering, Massachusetts Institute of Technology, Cambridge MA 02139 e-mail: ({tgibson@mit.edu})}
\thanks{ Z. Qu and A. M. Annaswamy are with the Department
of Mechanical Engineering, Massachusetts Institute of Technology, Cambridge,
MA, 02139.}
\thanks{E.~Lavretsky is with the Boeing Company, Huntington Beach CA, 92648.}}% <-this % stops a space
\begin{document}

\maketitle

\begin{abstract}
This note presents the design and analysis of an adaptive controller for a class of linear plants in the 
presence of output feedback. This controller makes use of a closed-loop reference model as an observer, and guarantees global stability and asymptotic output tracking. 
\end{abstract}

\section{Introduction}
While adaptive control has been studied since the 60's, the evolution of its use in real systems and the extent to which we fully understand its behavior has only been elucidated within the last decade. Stability of adaptive control systems came only in the 70's, with robustness and extensions to nonlinear systems coming in the 80's and 90's, respectively \cite{annbook,ioabook,kkkbook}. Recent directions in adaptive control pertain to guaranteed transient properties by using  a closed-loop architecture   for reference models \cite{gib13access,gib13acc1,gib13ecc,gib13acc2,eug10aiaa,lav12tac,lavbook,gibson_phd}. In this paper, we focus on linear  {\em Multi Input Multi Output} (MIMO) adaptive systems with partial state-feedback where we show that such closed-loop reference models can lead to a separation principle based adaptive controller which is simpler to implement compared to the classical ones in \cite{annbook,ioabook,kkkbook}. The simplification comes via the use of reference model states in the construction of the regressor, and not the classic approach where the regressor is constructed from filtered plant inputs and outputs.

%Within the last 4 years several researchers have strived to provide transient performance bounds for adaptive systems \cite{gib13acc1,gib13acc2,gib13ecc,gib13tranA,gib13access} and also attempted to provide methodologies for simplifying the design of output feedback adaptive controllers \cite{lavbook}. 
%
%In an attempt to improve the transient performance of adaptive systems the notion of a closed-loop reference model was introduced \cite{gib13acc1,gib13acc2,gib13ecc,gib13tranA,gib13access}. In addition to 

In general, the separation principle does not exist for nonlinear systems and few authors have analyzed it. Relevant work on the separation principle in adaptive control can be found in \cite{kha96,ata01}. The structures presented in \cite{kha96,ata01} are very generic, and as such, no global stability results are reported in this literature. Also, due to the generic nature of the results it is a priori assumed (or enforced through a saturation function) that the control input and adaptive update law are globally bounded functions with respect to the plant state \cite[Assumption 1.2]{ata01}. No such assumptions are needed in this work and the stability results are global.

The class of MIMO linear plants that we address in this paper satisfy two main assumptions. The first is that the number of outputs is greater than or equal to the number of inputs, and the second is that the first Markov Parameter has full column rank. The latter is equivalent to a relative degree unity condition in the {\em Single Input Single Output} (SISO) case. In addition to these two assumptions,  the commonly present assumption of stable transmission zeros is needed here as well. With these assumptions, an output feedback adaptive controller is designed that can guarantee stability and asymptotic tracking of the reference output. Unlike \cite{kha96,ata01}, no saturation is needed, and unlike \cite{eug10aiaa,lav12tac,lavbook}  asymptotic convergence of the tracking error to zero is proved for finite observer gains. Preliminary results on the control scheme presented in this work can be found in \cite{qu13aiaa}. An alternate approach using a linear matrix inequality was developed in \cite{dpw2015} and is successfully applied to a hypersonic vehicle model. An analytical approach was developed in \cite{max15} to handle a specific class of nonlinear uncertainties and achieves asymptotic convergence of the tracking error to zero with finite observer gains, and is shown to be applicable for a class of flexible aircraft platforms.

The paper is organized as follows. Section II states the control problem along with our assumptions. Section III proves stability for SISO and square MIMO systems. Section IV analyzes the use of an optimal observer in the design of the closed loop reference model as well as a methodology for extending the design to non-square MIMO systems. Section V contains a simulation example based on the longitudinal dynamics of an aircraft.  Conclusions are presented in Section VI.

\subsubsection*{Notation}
The 2-norm for vectors and the induced 2-norm for matrices is denoted as $\norm \cdot$. The differential operator is defined as $s= d/dt$ throughout. For a real matrix $A$, the notation $A^T$ is the matrix transpose. We use $I$ to denote the identity matrix. Big $O$-notation in terms of $\nu$ is presented as $O(\nu)$ and unless otherwise stated it is assumed that this holds for $\nu$ positive and sufficiently small. The definition of  {\em Strict Positive Real} (SPR), the {\em Kalman-Yacubovich-Popov} (KYP) Lemma, and the definition of transmission zero are given in Appendix \ref{app:lin}.

%Section II contains the preliminaries. Section III defines the control problem and contains two methods for designing the closed-loop reference model for square systems. Section IV presents the stability analysis for non-square systems. Section V contains the concluding remarks.

%This work presents a globally stable solution where no saturation functions are needed and furthermore our results hold for finite observer gains. This work also advances the results presented in \cite{eug10aiaa,lav12tac} where the model following error converged to a compact set that was inversely proportional to design parameters. In this note we are able to show that the modeling error asymptotically converges to zero for finite observer gains. We also make no a-priori assumptions regarding the boundedness of the regressor vector in relation to the model following error.
%%
%The main contribution of this paper are the asymptotic stability for a class of simple output feedback adaptive systems. The analysis presented here in is more generic than that presented in \cite{lavbook} and where as the results in \cite{lavbook} only prove boundedness, we are able to show asymptotic stability. We also remove the notion of a prescribed degree of freedom.
%

\section{Control Problem}

The class of plants to be  addressed in this paper is
\be\label{mod:eq:sys1}
\dot x = Ax+ B\Lambda u,\quad \quad y = C^T x
\ee
where $x\in \Re^n$, $u\in \Re^m$, and $y\in \Re^m$. $A$ and $\Lambda$ are unknown, but $B$ and $C$ are assumed to be known, and only  $y$ is assumed to be available for measurement. The goal is to design a control input $u$ so that $x$ tracks the closed-loop reference model state $x_m$ 
\be\label{mod:eq:ref}
\dot x_m = A_mx_m+ B r - L (y-y_m),\quad \quad y_m = C^T x_m
\ee
where $r\in \Re^m$ is the reference input and  and $L$ is a feedback gain that will be designed suitably. The reader is referred to references \cite{gib13acc1,gib13acc2,gib13ecc,gib13tranA,gib13access} for its motivation.

%We are only presenting results for square systems, i.e. the same number of outputs as inputs, in this paper. Note however that squaring-up and squaring-down procedures do exists that allow for the adaptive control of non-square systems, similar to what is done in \cite[\S 13]{lavbook}

The following assumptions are made throughout.
\begin{assm}\label{mo:as:1}
The product $C^TB$ is full rank.
\end{assm}
\begin{assm}\label{mo:as:2}
The pair $\{A_m, C^T\}$ is observable.% and the pair $\{A_m,B\}$ is controllable.
\end{assm}
%{\color{red}
%\begin{rem}
%For a system satisfying Assumptions 1 and 2, $\{A_m,B\}$ is necessarily controllable and thus it is also minimal.
%\end{rem}}
\begin{assm}\label{mo:as:3}
The system in \eqref{mod:eq:sys1} is minimum phase.\footnote{A MIMO system is  minimum phase if all of its transmission zeros are in the strict left half of the complex plane.}
\end{assm}
\begin{assm}\label{mo:as:4}
There exists a $\Theta^*\in \Re^{n\times m}$ such that $A+B\Lambda \Theta^{*T}= A_m$ and $K^*\in \Re^{m\times m}$ such that $\Lambda K^{*T} = I$.
\end{assm}
\begin{assm}\label{mo:as:5}
 $\Lambda$ is diagonal with positive elements.
\end{assm}
%\begin{rem}\label{rem:ctrb}
%Assumption 4 along with Assumption 5 implies that the pair $\{A,B\}$ is controllable as the matching condition requires both eigenvalue and eigenvector matching. 
%\end{rem}
\begin{assm}\label{mo:as:6} The uncertain matching parameter $\Theta^*$, and the input uncertainty matrix $\Lambda$ have a priori known upper bounds
\be\label{kb}
\bar \theta^* \triangleq \sup \norm{\Theta^*}  \text{ and }  \bar \lambda \triangleq  \sup \norm{\Lambda}.
\ee
\end{assm}
%{\color{red}
%\begin{assm} The triple (A,B,C) is minimal (I have not decided wether we need this yet). I think this is a natural consequence of Assumptions 1 and 2.
%\end{assm}}

\noindent Assumption \ref{mo:as:1} corresponds to one of the main assumptions mentioned in the introduction, and that is that the first Markov Parameter is nonsingular. The system in \eqref{mod:eq:sys1} is square and therefore the other main assumption mentioned in the introduction is implicitly satisfied. The extension to non-square systems is presented later in the text. Assumption 2 is necessary as our result requires the use of an observer like gain in the reference model, notice the $L$ in \eqref{mod:eq:ref}. Assumption 3 is common in adaptive systems as the KYP Lemma does not hold for plants with a right half plane transmission zero. 

Assumptions \ref{mo:as:4} and \ref{mo:as:5} imply that the pair $\{A,B\}$ is controllable, and are such that a matching condition is satisfied.  Such an assumption is commonly made in plants where states are accessible \cite{annbook}, but is introduced in this problem when only certain outputs are accessible. One application area where such an assumption is routinely satisfied is in the area of aircraft control \cite{lavbook}. Extensions of Assumption \ref{mo:as:4} to the case when the underlying regressor vector is globally Lipschitz are possible as well \cite{lavbook}. Assumption 5 can be relaxed to $\Lambda$ symmetric and full rank. 
%The sign of $\Lambda$ must be known so that we know in what direction our parameters must adapt. This is akin to knowing the sign of the high-frequency gain in classic SISO adaptive control \cite[\S 5,  \S 9.2]{annbook}. 
Assumption \ref{mo:as:6} facilitates an appropriate choice of $L$.
%This assumption is not needed in classic adaptive control. However we note the use of projection algorithms has become common practice in adaptive control as the best means for having robust adaptive laws. The use of a projection algorithm also requires Assumption 6 and thus if we are to use a projection algorithm to have a robust adaptive law then Assumption 6 would be implicit in the control design. 
The specifics of the control design are now addressed.

For the plant in \eqref{mod:eq:sys1} and \eqref{mod:eq:ref} satisfying the six assumptions above, we propose the following adaptive controller:
\be\label{eq:cin}
u=\Theta^T(t)x_m + K^T(t) r
\ee
\be\label{mod:eq:update0}
\begin{split}
\dot\Theta & =  - \Gamma_\theta x_m e_y^T  M 
\\
\dot K & =  -\Gamma_k r e_y^T M
\end{split}
\ee
where 
\be\label{M}
M \triangleq C^TB, \ee  $e_y=y-y_m$ and $\Gamma_\theta,\Gamma_k$ are both positive diagonal free design matrices. The matrix $M$ is referred to as the {\em mixing matrix} throughout.

The reason for the choice of the control input in \eqref{eq:cin} is simply because $x$ is not available for measurement, and the reference model state $x_m$ serves as an observer-state. Historically, the use of such an observer has always proved to be quite difficult, as the non-availability of the state proves to be a significant obstacle in determining a stable adaptive law. In the following, it is shown that these obstacles can be overcome for the specific class of multivariable plants that satisfy Assumptions 1 through \ref{mo:as:6}. 

From \eqref{mod:eq:sys1}, \eqref{mod:eq:ref}, and \eqref{eq:cin}, it is easy to show that the state error $e=x-x_m$ satisfies the dynamics
\be\begin{split}\label{mod:eq:e} \dot e &= (A_m+LC^T) e + B\Lambda ( \tilde\Theta^T x_m +  \tilde K^T r - \Theta^{*T} e) \\ e_y&=C^Te\end{split} \ee The structure of \eqref{mod:eq:e} and the adaptive laws suggest the use of the following Lyapunov function:
\be\label{disc:v00}
V=e^TPe + \text{Tr} (\Lambda \tilde\Theta^T \Gamma_\theta^{-1} \tilde\Theta) + \text{Tr} (\Lambda \tilde K^T \Gamma_k^{-1} \tilde K ) 
\ee
where for now it is assumed that $P=P^T>0$ satisfies the following equation
\be\begin{split}\label{disc:p00}
(A_m+LC^T)^TP+P(A_m+LC^T)&=-Q \\ PB&=CM
\end{split}\ee
where $Q=Q^T>0$.
Taking the derivative of \eqref{disc:v00} and using \eqref{mod:eq:update0}, \eqref{mod:eq:e}, and \eqref{disc:p00} it can be shown that
\be\label{star}
\dot V = -e^TQe + 2e^TPB \Lambda \Theta^{*T} e.
\ee
% highlighted equation here % Travis, check with projection, is there an inequality?%
Establishing sign-definiteness of $\dot V$ is therefore non-trivial as the size of the sign-indefinite term in \eqref{star} is directly proportional to the parametric uncertainty $\Theta^*$, and $P$ and $Q$ are necessarily correlated by \eqref{disc:p00}. In what follows, we will show how $L$ and $M$ can be chosen such that a $P$ and $Q$ satisfying \eqref{disc:p00} exist and furthermore, $\lim_{t\to\infty} e_y(t) =0$. It will be shown that stability for the above adaptive system can only be insured if $Q>0$ is sufficiently weighted along the $CC^T$ direction.

%has stable update laws which lead to $\lim_{t\to \infty} e_y(t)=0$. The model following error satisfies the following dynamical relation:
%\be\begin{split}\label{mod:eq:e} \dot e &= (A_m+LC^T) e + B\Lambda \left( \tilde\Theta^T x_m +  K^T r - \Theta^{*T} e\right) \\ e_y&=C^Te\end{split} \ee
%where $e=x-x_m$.

\section{Stability Analysis}

\subsection{Stability in the SISO Case}\label{mod:sec:siso}
The choice of $L$ is determined in two steps. First, an observer gain $L_s$ and mixing matrix $M$ are selected so that the transfer function $M^TC^T(sI-A-L_sC^T)B$ is {\em Strict Positive Real} (SPR).\footnote{$M$ is denoted the mixing matrix, as it mixes the outputs of $C^T(sI-A-L_sC^T)B$ so as to achieve strict positive realness.} Then the full observer gain $L$ is defined.

\begin{lem}\label{mod:lem:l1}
For  a SISO ($m=1$) system in \eqref{mod:eq:sys1} satisfying Assumptions \ref{mo:as:1}--\ref{mo:as:3}  there exists an $L_s$ such that
\be\label{min1}
C^T(sI-A_m-L_sC^T)^{-1}B = \frac{a}{s+\rho}
\ee
where $\rho>0$ is arbitrary and $a = C^TB$.
\end{lem}
\begin{proof}
Given that $C^TB$ is non-zero $C^T(sI-A_m-L_sC^T)^{-1}B$ is a relative degree one transfer function. In order to see this fact, consider a system in control canonical form, and compute the coefficient for $s^{n-1}$ in the numerator. By Assumption 2, all zeros of the transfer function $C^T(sI-A)^{-1}B$ are stable, and since zeros are invariant under feedback, $C^T(sI-A_m)^{-1}B$ is minimum phase as well. Assumption 2 implies that the eigenvalues of $A_m+L_sC^T$ can be chosen arbitrarily. Therefore, one can place $n-1$ of the eigenvalues of $A_m+L_sC^T$ at the $n-1$ zeros of $C^T(sI-A_m)^{-1}B$ and its $n$-th eigenvalue clearly at $-\rho$. 
\end{proof}

%\begin{rem}
%The Ackerman algorithm for pole placement of SISO systems can be found in \cite[pg. 201]{kai80}.
%\end{rem}

The choice of $L_s$ in Lemma \ref{mod:lem:l1} results in a relative degree one transfer function with a single pole not canceling the zeros. This system however need not be SPR as $a$ may be negative; however $\frac{a^2}{s+\rho}$ is SPR and thus the following Corollary holds.

\begin{cor}\label{corM}
If $L_s$ is chosen as in \eqref{min1} and $M$ selected as in \eqref{M},
the SISO transfer function $M^TC^T(sI-A_m-L_sC^T)^{-1}B$  is SPR. Therefore, there exists $P=P^T>0$ and $Q_s=Q_s^T>0$ such that
\be\begin{split}\label{eq:ly2}
(A_m+L_sC^T)^TP+P(A_m+L_sC^T)&=-Q_s \\ PB&=CM.
\end{split}\ee
\end{cor}

\begin{lem}\label{lem:proof1}
Choosing $L=L_s-\rho BM^T$ where $L_s$ is defined in Lemma \ref{mod:lem:l1} and $\rho>0$ is arbitrary, the transfer function $M^TC^T(sI-A_m-LC^T)^{-1}B$ is SPR and satisfies:
\be
\begin{split}
\label{eq:ly2p}
(A_m+LC^T)^TP+P(A_m+LC^T)  = -Q \\
 Q   \triangleq Q_s  + 2\rho  C MM^TC^T 
\end{split}
\ee
where $P$ and $Q_s$ are defined in \eqref{eq:ly2} and $M$ is defined in \eqref{M}.
\end{lem}
\begin{proof}
Starting with the first equation in \eqref{eq:ly2} and adding the term $-\rho\left(PBM^TC^T+CMB^TP\right)$
on both sides of the inequality results in the following equality
\ben\begin{split}
(A_m+ LC^T)^TP&+ P(A_m+LC^T) =  \\&  -Q_s-\rho\left(PB M^T C^T+C M B^TP\right).
\end{split} \een
Using the second equality in \eqref{eq:ly2} the above equality simplifies to \eqref{eq:ly2p}
%\ben
%(A_m+LC^T)^TP+ P(A_m+LC^T)=  -Q_s-2 \rho   CMM^TC^T.\een 
\end{proof}

%
%From \eqref{min1}, the state variable representation of the error dynamics can be written in a minimal form
%\be\label{mod:eq:e2} 
% \dot e_y = -\mu e_y + C^TB \Lambda \left(\tilde\Theta^T x_m +  \tilde K^T r - \Theta^{*T} e\right) 
%\ee

%Since $P$ satisfies \eqref{eq:ly2p}, we choose $M$ in \eqref{mod:eq:update0} to be $C^TB$ so that the update laws reduce to.

%\be\label{mod:eq:update}
%\begin{split}
%\dot\Theta = &- \Gamma_\theta x_m e_y  B^TC 
%\\
%\dot K = & -\Gamma_k r e_y B^TC
%\end{split}
%\ee

\begin{thm}\label{thm1}
The closed-loop adaptive system specified by  \eqref{mod:eq:sys1}, \eqref{mod:eq:ref}, \eqref{eq:cin} and \eqref{mod:eq:update0}, satisfying assumptions 1 to \ref{mo:as:6}, with $L$ as in Lemma \ref{lem:proof1},  $M$ chosen as in \eqref{M}, and $\rho > \rho^*$  has globally bounded solutions with $\lim_{t\to\infty}e_y(t)=0$ with \be\label{rs} \rho^* = \frac{ \bar \lambda^2{\bar \theta}^{*2}}{2\lambda_{\text{min}}(Q_s)},\ee
where  $\bar\lambda$ and $\bar\theta^*$ are a priori known bounds defined in \eqref{kb}.
\end{thm}

\begin{proof}
We choose the lyapunov candidate \eqref{disc:v00} where $P$ is the solution to \eqref{eq:ly2} and satisfies \eqref{eq:ly2p}. Taking the time derivative of \eqref{disc:v00} along the system trajectories in \eqref{mod:eq:e},  and using the relations in \eqref{eq:ly2}, \eqref{eq:ly2p}, and \eqref{mod:eq:update0}, the following holds:
\be\label{mod:1}
\begin{split}
\dot V =& - e^T(Q+2 \rho  CMM^TC^T)e - 2e^TPB\Lambda \Theta^{*T} e \\ & + 2e^TPB \Lambda \tilde\Theta^T x_m + 2 \text{Tr} (\Lambda \tilde\Theta^T  x_m e_y^T M)  \\ &+ 2e^TPB \Lambda \tilde K^T r + 2 \text{Tr} (\Lambda \tilde K^T r e_y^T M)
\end{split}
\ee
Using the fact that $PB=CM$ from \eqref{eq:ly2} and the fact the Trace operator is invariant under cyclic permutations the inequality in \eqref{mod:1} can be rewritten as
\be\label{mod:2}
\begin{split}
\dot V =& - e^T(Q+2\rho CMM^TC^T)e - 2e^TCM\Lambda \Theta^{*T} e \\ & + 2e^TC M \Lambda \tilde\Theta^T x_m - 2 e_y^TM \Lambda \tilde\Theta^T  x_m  \\ &+ 2e^TCM \Lambda \tilde K^T r - 2  e_y^TM \Lambda \tilde K^T r 
\end{split}
\ee
Using the fact that $e_y = C^T e$, the 2nd and 3rd lines in the above equation equal zero. Therefore, \eqref{mod:2} can be written as $\dot V = -  \mathcal E^T \mathcal Q(\rho) \mathcal E$
where
\ben
 \mathcal Q(\rho)= \bb 2\rho MM^T &  M \Lambda \Theta^{*T} \\  \Theta^{*}\Lambda M^T & Q_s\eb \quad  
 \mathcal E= \bb e_y \\ e 
 \eb. \een
Given that  ${\rho>\rho^*>0}$, $ 2M\rho M^T-  M \Lambda\Theta^{*T} Q_s^{-1}\Theta^{*}\Lambda M^T>0$ by \eqref{rs} and $Q_s$ is posititve definite by design.   By Schur complement, $\mathcal Q(\rho)$ is positive definite. Therefore $\dot V \leq 0$ and thus $e_y,e,\tilde\Theta,\tilde K\in \mathcal L_\infty$. Furthermore, given that $M$ is positive definite ${e_y\in \mathcal L_2}$. Using Barbalat Lemma it follows that $\lim_{t\to\infty} e_y(t) = 0$.
\end{proof}

\begin{rem}
Theorem 1 implies that a controller as in \eqref{eq:cin} with the state replaced by the observer state $x_m$ will guarantee stability, thereby illustrating that the separation principle based adaptive control design can be satisfactorily deployed. It should be noted however that two key parameters $L$ and $M$ had to be suitably chosen. If $L = L_s$ then stability is not guaranteed. That is, simply satisfying an SPR condition is not sufficient for stability to hold. It is imperative that $Q$ be chosen as in \eqref{eq:ly2p}, i.e. be sufficiently positive
along the output direction $CC^T$ so as to contend with the sign indefinite term $2e^TPB \Lambda \Theta^{*T} e$ in $\dot V$. The result does not require that $L_s$ be chosen so that perfect pole zero cancellation occurs in Lemma \ref{mod:lem:l1}, all that is necessary is that the phase lag of $C^T(j\omega I-A_m-L_sC^T)^{-1}B$ never exceeds 90 degrees. Finally, it should be noted that any finite $\rho > \rho^*$ ensures stability.
\end{rem}

%\begin{rem}
%Note that if ${L=L_s}$ then stability is not guaranteed. Simply satisfying an SPR condition is not sufficient for stability to hold. It is imperative that $Q$ be sufficiently positive along the output direction $CC^T$ so as to contend with the sign indefinite term $2e^TPB \Lambda \Theta^{*T} e$ in $\dot V$. This concisely illustrates the separation principle. Through the selection of $\rho$ stability is guaranteed and the eigenvalues of the closed-loop reference model are appropriately placed. It is also important to note that all terms in $\rho^*$ are independent of $\rho$, and thus the problem of guaranteeing stability by $\rho$ sufficiently large is well posed.
%\end{rem}

\subsection{Stability in the MIMO Case}\label{mod:sec:mimo} \label{sec:mimo1}
Stability in the MIMO case follows the same set of steps as in the SISO case. First, an $L_s$ and $M$  are defined such that the transfer function $M^TC^T(sI-A_m-L_sC^T)B$ is SPR. Then $L$ is defined such that the underlying adaptive system is stable. The following Lemmas mirror the results from Corollary \ref{corM} and Lemma \ref{lem:proof1}.

\begin{lem}\label{LEM2}
For the MIMO system in \eqref{mod:eq:sys1} satisfying Assumptions \ref{mo:as:1}--\ref{mo:as:3}  with $M$ chosen as in \eqref{M} there always exists an $L_s$ such that
$M^TC^T(sI-A_m-L_sC^T)^{-1}B$
is SPR.
\end{lem}
\begin{proof}
An algorithm for the existence and selection of such an $L_s$ is given in \cite{yu10}.\end{proof}

\begin{rem}
 In order to apply the results from \cite{yu10}, the MIMO system of interest must be 1) minimum phase and 2) $M^TC^TB$ must be symmetric positive definite. By Assumption \ref{mo:as:3}, $C^T(sI-A)^{-1}B$ is minimum phase, and therefore $C^T(sI-A_m)^{-1}B$ is minimum phase as well. Also, given that $M$ is full rank, the transmission zeros of $C^T(sI-A_m)^{-1}B$ are equivalent to the transmission zeros of $M^TC^T(sI-A_m)^{-1}B$, see Lemma \ref{lem:inv} in Appendix \ref{app:lin}. Therefore, condition 1 of this remark is satisfied. We now move on to condition 2.
 
By Assumption \ref{mo:as:1} $C^TB$ is full rank, and by the definition of  $M$ in \eqref{M} it follows that $M^TC^TB=B^TCM>0$, which is a necessary condition for $MC^T(sI-A_m)^{-1}B$ to be SPR, see Corollary \ref{cor:spr} in Appendix \ref{app:lin}. A similar explicit construction of an $L_s$ such that $M^TC^T(sI-A_m-L_sC^T)^{-1}B$ is SPR can be found in \cite{huang}.
\end{rem}

\begin{lem}\label{LEM3}
Choosing $L=L_s-\rho BM^T$ where $L_s$ is defined in Lemma \ref{LEM2} and $\rho>0$ is arbitrary, the transfer function $M^TC^T(sI-A_m-LC^T)^{-1}B$ is SPR and satisfies:
\be
\begin{split}
\label{eq:ly3p}
(A_m+LC^T)^TP+P(A_m+LC^T)  = -Q \\
 Q   \triangleq Q_s  + 2\rho  C MM^TC^T \\
 PB=CM 
\end{split}
\ee
where $P=P^T>0$ and $Q_s=Q_s^T>0$ are independent of $\rho$ and $M$ is defined in \eqref{M}.
\end{lem}
\begin{proof}
This follows the same steps as in the proof of Lemma \ref{lem:proof1}.\end{proof}

\begin{thm}
The closed-loop adaptive system specified by  \eqref{mod:eq:sys1}, \eqref{mod:eq:ref}, \eqref{eq:cin} and \eqref{mod:eq:update0}, satisfying assumptions 1 to \ref{mo:as:6}, with $L$ as in Lemma \ref{LEM3},  $M$ chosen as in \eqref{M}, and $\rho > \rho^*$  has globally bounded solutions with $\lim_{t\to\infty}e_y(t)=0$ where $\rho^*$ is defined in \eqref{rs}. \end{thm}
\begin{proof}
This follows the same steps as in the proof of Theorem \ref{thm1}.
\end{proof}

\section{Extensions}

%\subsection{MIMO LQG/LTR} \label{sec:mimo2}
In the previous section a method was presented for choosing $L$ in \eqref{mod:eq:ref} and $M$ in \eqref{mod:eq:update0} so that the overall adaptive system is  stable and $\lim_{t\to\infty} e(t) = 0$. For the SISO and MIMO cases the proposed method, thus far, is  a two step process. First a feedback gain and mixing matrix are chosen such that a specific transfer function is SPR. Then, the feedback gain in the first step is augmented with an additional feedback term of sufficient magnitude along the direction $BM^T$ so that stability of the underlying adaptive system can be guaranteed. 

In this section,  the method is extended to two different cases. In the first case, we apply this method to an LQG/LTR approach proposed in  \cite{lavbook} and show that asymptotic stability can be derived thereby extending the results of  \cite{lavbook}. In the second case, the method is extended to non-square plants.

\subsection{MIMO LQG/LTR} \label{sec:mimo2}
The authors in \cite{lavbook} suggested using an LQG approach for the selections of $L$ and $M$, motivated by the fact the underlying observer (which coincides with the closed-loop reference model as shown in \eqref{mod:eq:ref}) readily permits the use of such an approach and makes the design more in line with the classical optimal control approach.

%Since the reference model is essentially acting like an observer in this construction, the authors in \cite{lavbook} suggested using an LQG approach for the selection of $L$ and the mixing matrix $M$. This makes the selection of $L$  more in line with a classical linear systems approach to observer design. 
In \cite{lavbook} the proposed method is only shown to be stable for finite $L$, where as in this section it is show that in fact $\lim_{t\to\infty} e(t) =0$. Furthermore, we note that the prescribed degree of stability as suggested in \cite[Equation 14.26]{lavbook} through the selection of $\eta$ is in fact not needed. The analysis below shows that stability is guaranteed due to sufficient weighting of the underlying $Q$ matrix along the $CC^T$ direction.

Let $L$ in \eqref{mod:eq:ref} be chosen as \cite{lavbook}
\be\label{mod:l22}
L= L_\nu \triangleq	-P_\nu C R_\nu^{-1}.
\ee
where $P_\nu$ is the solution to the Riccati Equation
\be\label{mod:11}
P_\nu A_m^T + A_mP_\nu -  P_\nu C R_\nu^{-1} C^T P_\nu + Q_\nu=0
\ee
where $Q_0=Q_0^T>0$ in $\Re^n$ and $R_0=R_0^T>0$ in $\Re^m$ and $\nu>0$, with 
$Q_\nu = Q_0+\left(1+\frac{1}{\nu}\right)BB^T$ and $ R_\nu = \frac{\nu}{\nu+1}R_0$.
Note that \eqref{mod:11} can also be represented as
%\be
%P_\nu (A_m+L_\nu C^T)^T + (A_m+L_\nu C^T)P_\nu =- P_\nu C R_\nu^{-1} C^T P_\nu - Q_\nu
%\ee
%which is also equivalent to
\be\label{ric2}
A_\nu^T\tilde P_\nu  + \tilde P_\nu  A_\nu = -C R_\nu^{-1} C^T  - \tilde  Q_\nu
\ee
where
$A_\nu = A_m+L_\nu C^T$, $\tilde P_\nu = P_\nu^{-1} \text{ and } \tilde Q_\nu = \tilde P_\nu Q_\nu \tilde P_\nu.$ Given that our system is observable and $Q$ and $R$ are symmetric and positive definite, the Riccati equation has a solution $P_\nu$ for all fixed $\nu$. We are particularly interested in the limiting solution when $\nu$ tends to zero. The Riccati equation in \eqref{mod:11} is very similar to those studied in the LTR literature, with one very significant difference. In LTR methods the state weighting matrix is independent of $\nu$ where as in our application $Q_\nu$ tends to infinity for small $\nu$. 

%The existence of a limiting solution is explored in detail in \cite{gib_optimal_arx} and will only detract from the discussion if presented here. 
%The existence of a limiting solution leads to the following Lemma which is similar to \cite[Corollary 13.1]{lavbook}.

%\begin{lem} \label{existP0}
%If Assumptions 1, 2, and 3 are satisfied then $\lim_{\nu\to 0}\nu P_\nu =0$, $\lim_{\nu\to 0} P_\nu =P_0$ where $0<P_0^T=P_0<\infty$, and the following asymptotic relation holds
%\be\label{Pvassymptotic}
%P_\nu = P_0 + P_1 \nu + O(\nu^2).
%\ee
%\end{lem}
%\begin{proof} $P_\nu$ is the solution to a dual of that which is presented in Appendix \ref{}. By Theorem \ref{} and corollary $\lim_{\nu\downarrow0}P_\nu$ is well defined. {\color{red}Given that $P_\nu$ is continuous in $\nu$ the asymptotic relation in \eqref{Pvassymptotic} holds as well.}
%\end{proof}

\begin{lem}
If Assumptions 1 through 5 are satisfied then
$\lim_{\nu\to 0}\nu P_\nu =0$, $\lim_{\nu\to 0} P_\nu =P_0$ where $0<P_0^T=P_0<\infty$, and the following asymptotic relation holds
\be\label{Pvassymptotic}
P_\nu = P_0 + P_1 \nu + O(\nu^2).
\ee
Furthermore, there exists a unitary matrix $W\in \Re^{m \times m}$ such that
\be
\label{eq:W}
P_0C = B W^T \sqrt{R_0}, \quad \text{and} \quad 
\tilde P_0B  = C R_0^{-1/2}W
\ee
where $\tilde P_0=P_0^{-1}$ and $W=(UV)^T$ with $B^TCR_0^{-1/2} = U \Sigma V$. Finally, the inverse $\tilde P_\nu\triangleq P_\nu^{-1}$ is well defined in limit of small $\nu$ and 
\be
\tilde P_\nu = \tilde P_0 + \tilde P_1 \nu + O(\nu^2).
\ee
\end{lem}
A full proof of this result is omitted to save space. The  following two facts, 1) $\lim_{\nu\to 0}\nu P_\nu =0$, and 2) $\lim_{\nu\to 0} P_\nu =P_0$ where $0<P_0^T=P_0<\infty$ follow by analyzing the integral cost 
\ben
x^T(0)P_\nu x(0) = \min \int_0^\infty x^T(\tau) Q_\nu x(\tau) + u^T(\tau) R_\nu u(\tau) \ d\tau
\een
in the same spirit as was done in \cite{kwa:tac72}. In order to apply the results from \cite{kwa:tac72} the system must be observable (Assumption 2), controllable (Assumptions 4 and 5), minimum phase (Assumption 3), and  $C^TB$ must be full rank (Assumption 1). For a detailed analysis of the asymptotic expansions $P_\nu = P_0 + P_1 \nu + O(\nu^2)$ and $\tilde P_\nu = \tilde P_0 + \tilde P_1 \nu + O(\nu^2)$ see \cite[\S 13.3, Theorem 13.2, Corollary 13.1]{lavbook}.

The update law for the adaptive parameters is then given as 
\be\label{mod:eq:update3}
\begin{split}
\dot\Theta = &- \Gamma_\theta x_m e_y^T  R_0^{-1/2}W
\\
\dot K = &-\Gamma_k r e_y^T R_0^{-1/2}W
\end{split}
\ee
where $W$ is defined just below \eqref{eq:W}.

\begin{thm}\label{thm:emt}The closed-loop adaptive system specified by  \eqref{mod:eq:sys1}, \eqref{mod:eq:ref}, \eqref{eq:cin} and \eqref{mod:eq:update3}, satisfying assumptions 1 to 6, with $L$ as in \eqref{mod:l22}, and $\nu$ sufficiently small has globally bounded solutions with $\lim_{t\to\infty}e_y(t)=0$.
\end{thm}
\begin{proof}
Consider the Lyapunov candidate $V=e^T\tilde P_0 e +  \text{Tr} (\Lambda \tilde\Theta^T \Gamma_\theta^{-1} \tilde\Theta) + \text{Tr} (\Lambda \tilde K^T \Gamma_k^{-1} \tilde K )$.
Taking the derivative along the system trajectories and substitution of the update laws in \eqref{mod:eq:update3} results in
\be\label{eq:LE}
\begin{split}
\dot V = &  e^TA_\nu^T\tilde P_0 e + e^T\tilde P_0 A_\nu e -2e^T\tilde P_0 B\Lambda \Theta^{*T} e 
 \\ & + 2e^T\tilde P_0 B \Lambda \tilde\Theta^T x_m + 2 \text{Tr} (\Lambda \tilde\Theta^T  x_m e_y^T R_0^{-1/2}W)  \\ &+ 2e^T\tilde P_0 B \Lambda \tilde K^T r + 2 \text{Tr} (\Lambda \tilde K^T r e_y^T R_o^{-1/2} W). 
\end{split}
\ee
The first step in the analysis of the above expression is to replace the elements $A_\nu^T\tilde P_0$ and $ \tilde P_0 A_\nu$ with bounds in terms of  $A_\nu^T\tilde P_\nu$ and $ \tilde P_\nu A_\nu$. First note that the following expansions hold in the limit of small $\nu$
\ben
\begin{split}
A_\nu^T\tilde P_\nu &=A_\nu^T\tilde P_0 + \nu A_\nu^T\tilde P_1+ O(\nu) \\
\tilde P_\nu A_\nu & = \tilde P_0 A_\nu + \nu \tilde P_1A_\nu+ O(\nu),
\end{split}
\een
where we have simply expanded the term $\tilde P_\nu$.
Expanding $A_\nu$ as $A_m -P_\nu C R_0^{-1} C^T \frac{\nu+1}{\nu} $ , the above relation simplifies to the following asymptotic relation as $\nu$ approaches 0,
\be\label{eq:ass123}
\begin{split}
A_\nu^T\tilde P_\nu &=A_\nu^T\tilde P_0   - CR_0^{-1}C^T  P_\nu \tilde P_1 + O(\nu)\\
\tilde P_\nu A_\nu& = \tilde P_0 A_\nu -  \tilde P_1  P_\nu C R_0^{-1} C^T+ O(\nu)
\end{split}
\ee
Substitution of \eqref{eq:ass123} for the expressions $A_\nu^T\tilde P_0$ and $ \tilde P_0 A_\nu$ in \eqref{eq:LE} results in the following inequality
\be
\begin{split}
\dot V \leq &  e^TA_\nu^T\tilde P_\nu e + e^T\tilde P_\nu A_\nu e -2e^T\tilde P_0 B\Lambda \Theta^{*T} e \\
&  + e^TCR_0^{-1}C^T  P_\nu \tilde P_1e +    e^T\tilde P_1 P_\nu C R_0^{-1} C^Te + O(\nu) e^Te 
 \\ & + 2e^T\tilde P_0 B \Lambda \tilde\Theta^T x_m + 2 \text{Tr} (\Lambda \tilde\Theta^T  x_m e_y^T R_0^{-1/2}W)  \\ &+ 2e^T\tilde P_0 B \Lambda \tilde K^T r + 2 \text{Tr} (\Lambda \tilde K^T r e_y^T R_o^{-1/2} W). 
\end{split}
\ee
%
%Adding and subtracting the terms $
%e^T  (A_m+L_\nu C^T)^T \sum_{i=1}^\infty \nu^i \tilde P_ie + e^T \sum_{i=1}^\infty \nu^i \tilde P_i  (A_m+L_\nu C^T)  e $ and noting that $\tilde P_\nu$ satisfies the following  asymptotic expansion $\tilde P_\nu = \sum_{i=0}^\infty \tilde P_i \nu^i $,
%\ben
%\begin{split}
%\dot V =  &  e^T (A_m+L_\nu C^T)^T\tilde P_\nu e + e^T\tilde P_\nu (A_m+L_\nu C^T)e \\& + 2e^T\tilde P_0 B\Lambda \Theta^{*T} e \\ & - e^T  (A_m+L_\nu C^T)^T \sum_{i=1}^\infty \nu^i \tilde P_i e  -e^T   \sum_{i=1}^\infty \nu^i \tilde P_i  (A_m+L_\nu C^T) e
% \\ & + 2e^T\tilde P_0 B \Lambda \tilde\Theta^T x_m + 2 \text{Tr} (\Lambda \tilde\Theta^T  x_m e_y^T R_0^{-1/2}W)  \\ &+ 2e^T\tilde P_0 B \Lambda \tilde K^T r + 2 \text{Tr} (\Lambda \tilde K^T r e_y^T R_o^{-1/2} W). 
%\end{split}
%\een
Substitution of \eqref{ric2} in to the first line above,
%\ben
%\begin{split}
%\dot V =  & - e^T\tilde Q_\nu e +e ^T C R_\nu^{-1} C^T e - 2e^T\tilde P_0 B\Lambda \Theta^{*T} e \\ & - e^T  (A_m+L_\nu C^T)^T \sum_{i=1}^\infty \nu^i \tilde P_i e \\ & -e^T   \sum_{i=1}^\infty \nu^i \tilde P_i  (A_m+L_\nu C^T) e
% \\ & + 2e^T\tilde P_0 B \Lambda \tilde\Theta^T x_m + 2 \text{Tr} (\Lambda \tilde\Theta^T  x_m e_y^T R_0^{-1/2}W)  \\ &+ 2e^T\tilde P_0 B \Lambda \tilde K^T r + 2 \text{Tr} (\Lambda \tilde K^T r e_y^T R_o^{-1/2} W). 
%\end{split}
%\een
and using the fact that $\tilde P_0 B = CR_0^{-1/2} W $ for the expressions in the bottom two lines,
%\be
%\begin{split}
%\dot V = & - e^T\tilde Q_\nu e  - e^T C \frac{\nu+1}{\nu}R_0^{-1} C^T e \\ & -e^T  (A_m-\tilde P_\nu C \frac{\nu+1}{\nu}R_0^{-1} C^T)^T \sum_{i=1}^\infty \nu^i \tilde P_i e \\ & -e^T   \sum_{i=1}^\infty \nu^i \tilde P_i  (A_m- \tilde P_\nu C \frac{\nu+1}{\nu}R_0^{-1} C^T) e
%\\ & - 2e^T C R_0^{-1/2}W\Lambda \Theta^{*T} e. 
%\end{split}
%\ee
%collecting the $O(\nu)$ terms  $(A_m+L_\nu C^T)^T \sum_{i=1}^\infty \nu^i \tilde P_i $ and $ \sum_{i=1}^\infty \nu^i \tilde P_i (A_m+L_\nu C^T) $, and noting that $1+\frac 1 \nu \geq \frac 1 \nu$ results in
\ben
\begin{split}
\dot V \leq & - e^T\tilde Q_\nu e  - \frac{\nu+1}{\nu} e_y^T R_0^{-1} e_y  + O(\nu) e^Te\\ & +e^TCR_0^{-1}C^T  P_\nu \tilde P_1e +    e^T\tilde P_1 P_\nu C R_0^{-1} C^Te\\ & - 2e^TC R_0^{-1/2}W\Lambda \Theta^{*T} e .
\end{split}
\een
Using the fact that $e_y=C^T e$ and $\nu+1\geq 1$, the following inequality holds for $\nu$ sufficiently small
\ben
\begin{split}
\dot V \leq & - e^T\tilde Q_\nu e  - \frac{1}{\nu} e_y^T R_0^{-1} e_y  + O(\nu) e^Te \\ & + e_y^T R_0^{-1} C^T P_\nu  \tilde P_1 e  +e^T   \tilde P_1   P_\nu C R_0^{-1} e_y\\ & 
- 2e_y^T R_0^{-1/2}W \Theta^{*T} e. 
\end{split}
\een
Expanding $P_\nu$ in the second line above
\be
\begin{split}
\dot V \leq & - e^T\tilde Q_\nu e  - \frac{1}{\nu} e_y^T R_0^{-1} e_y  + O(\nu) e^Te \\ & + e_y^T R_0^{-1} C^T P_0  \tilde P_1 e  +e^T   \tilde P_1  P_0 C R_0^{-1} e_y\\ & 
- 2e_y^T R_0^{-1/2}W \Theta^{*T} e. 
\end{split}
\ee
 
Let $P_\Theta \triangleq-R_0^{-1} C^T P_0  \tilde P_1 +  R_0^{-1/2}W \Theta^{*T} $, then the above inequality can be simplified as $\dot V \leq -  \mathcal E^T \mathcal Q(\nu) \mathcal E + O(\nu) e^Te$
where
\be
 \mathcal Q(\nu) = \bb \frac{1}{\nu}R_0^{-1} & P_\Theta \\ P_\Theta^T & \tilde Q_\nu \eb\text{ and } \mathcal E= \bb e_y \\ e \eb .
\ee
Note that $P_\Theta$ is independent of $\nu$ and  ${ \lim_{\nu \to 0} \tilde Q_\nu \geq \tilde P_0 Q_0 \tilde P_0>0}$. Thus for $\nu$ sufficiently small $ \frac{1}{\nu}R_0^{-1} - P_\Theta \tilde Q_\nu^{-1} P_\Theta^T>0$. Therefore $\mathcal Q(\nu)$ is positive definite and for $\nu$ sufficiently small $\mathcal Q(\nu)-O(\nu)I>0$ as well, where $I$ is the identity matrix. Thus the adaptive system is bounded for sufficiently small $\nu$. As before, it follows that $e_y\in \mathcal L_2$, and by Barbalat Lemma, $\lim_{t\to\infty}e_y(t)=0$.
\end{proof}
\begin{rem}
The same discussion for the SISO and MIMO cases is valid for the LQG/LTR based selection of $L$. Stability follows do to the fact that the Lyapunov candidate suitably includes the ``fast dynamics'' along the $e_y$ error dynamics. This fact is illustrated in \eqref{ric2} with the term $ C R_\nu^{-1} C^T$ appearing on the right hand, which when expanded in terms of $\nu$ takes the form $\frac{1+\nu}{\nu}C R_0^{-1} C^T$. By directly comparing $\frac{1+\nu}{\nu}C R_0^{-1} C^T$ to the term $2 \rho CMM^TC^T$ on the right hand side of \eqref{eq:ly2p}, increasing $\rho$ and decreasing $\nu$ have the same affect on the underlying Lyapunov equations. Thus, stability is guaranteed so long as $\rho$ is sufficiently large or equivalently, $\nu$ sufficiently small.
\end{rem}

\begin{rem}The stability analysis of this method was first presented in \cite{lavbook}.  This remark illustrates why the stability analysis presented in \cite{lavbook} resulted in $e(t)$ converging to a compact set for finite $\nu$. Consider the Lyapunov candidate from \cite[(14.43)]{lavbook} repeated here in \ben V=e^T\tilde P_\nu e +  \text{Tr} (\Lambda \tilde\Theta^T \Gamma^{-1} \tilde\Theta) + \text{Tr} (\Lambda \tilde K^T \Gamma^{-1} \tilde K ).\een
Taking the time derivative along the system trajectories
\ben
\begin{split}
\dot V =& - e^T\tilde Q_\nu e  - e^T C R_\nu^{-1} C^T e + 2e^T\tilde P_\nu B\Lambda \Theta^{*T} e \\ & + 2e^T\tilde P_\nu B \Lambda \tilde\Theta^T x_m + 2 \text{Tr} (\Lambda \tilde\Theta^T  x_m e_y^T R_0^{-1/2}W)  \\ &+ 2e^T\tilde P_\nu B \Lambda \tilde K^T r + 2 \text{Tr} (\Lambda \tilde K^T r e_y^T R_o^{-1/2} W) 
\end{split}
\een
which can be simplified to
\ben
\begin{split}
\dot V \leq& - e^T\tilde Q_\nu e  - e^T C R_\nu^{-1} C^T e + 2e^T\tilde P_\nu B\Lambda \Theta^{*T} e \\ & + O(\nu) \norm{e}\norm{x_m} + O(\nu) \norm{e}\norm r
\end{split}
\een
as $\nu \to 0$.
Note that $x_m$ is a function of $e$. Therefore, it is difficult to bound $x_m$ before the boundedness of $e$ is obtained. Furthermore, the presence of $r(t)$ on the righthand side will always perturb $V$ away from $0$ for all finite $\nu$. In Theorem \ref{thm:emt} we overcame this issue by selecting a slightly different Lyapunov function, $\tilde P_\nu$ was replaced by the limiting solution of $\tilde P_0$. It would appear to be a rather benign change to the Lyapunov candidate. This change however allows us to go from stability to the model following error converging to zero.
\end{rem}

\subsection{Extension to Non-square Systems}
%{\color{red}Add references to dan's and max's new work after my meeting}
Consider dynamics of the following
form
\begin{equation}
\dot{x}=Ax+B_{1}\Lambda u,\qquad y=C^{T}x\label{eq:non-sqaure system}
\end{equation}
where $x\in\Re^{n}$, $u\in\Re^{m}$, $y\in\Re^{p}$ and $p>m$. $B_{1}\in\Re^{n\times m}$
and $C\in\Re^{n\times p}$ are known. $A\in\Re^{n\times n}$ and $\Lambda\in\Re^{m\times m}$
are unknown. To address the non-square aspect Assumption 1 is replaced with the following:
\begin{assm}
Rank$(C)=p$ and  Rank$(C^TB_1)=m$.
\end{assm}
Again, the goal is to design a controller such that $x(t)$ follows the reference model:
\begin{equation}
\dot{x}_{m}=A_{m}x_{m}+B_{1}r-Le_y,\qquad y_m=C^{T}x_{m}\label{eq:ref system}
\end{equation}
where $C^{T}(sI-A_{m})^{-1}B_{1}$ represents the ideal behavior
responding to a command $r$. 

%
%Previous literature has presented a method to choose an $L$ such that stable output-feedback adaptive
%control is possible \cite{misra1}\cite{RefWorks:90}. The first step
%is to slightly modify the system so that it becomes square and relative
%degree one. The modified system represents a virtual system and is
%only used for adaptive control design. 

\begin{lem}
\label{lem:B} For a non-square system in the form of (\ref{eq:non-sqaure system})
and (\ref{eq:ref system}) that satisfies Assumptions 2, 3, and 7, there
exists a $B_{2}\in\mathbb{\Re}^{n\times(p-m)}$ such that the ``squared-up''
system $C^{T}(sI-A_{m})^{-1}B\label{eq:Wsu}$
is minimum phase, and $C^TB$ is full rank, where \be\label{Bdef} B=\bb B_{1} &  B_{2}\eb.\ee\end{lem}
\begin{proof}
The reader is referred to \cite{qu13} for further details.
\end{proof}

We now consider the squared-up plant $\{A_m,B,C^T\}$ and state the lemmas corresponding to Lemma \ref{LEM2} and Lemma \ref{LEM3}.

\begin{lem}\label{LEM2s}
For the MIMO system in \eqref{eq:non-sqaure system} satisfying Assumptions \ref{mo:as:2}, \ref{mo:as:3} and 7  with $M$ chosen as in \eqref{M} with $B$ as defined in \eqref{Bdef} there exists an $L_s$ such that
$M^TC^T(sI-A_m-L_sC^T)^{-1}B$
is SPR.
\end{lem}

%\begin{rem}
%
%In order to apply the results from \cite{yu10}, the system of interest must be 1) minimum phase and 2) $M^TC^TB$ must be symmetric positive definite. By Assumption \ref{mo:as:3}, $C^T(sI-A)^{-1}B$ is minimum phase, and therefore $C^T(sI-A_m)^TB$ is minimum phase as well. Also, given that $M$ is full rank, the transmission zeros of $C^T(sI-A_m)^{-1}B$ are equivalent to the transmission zeros of $M^TC^T(sI-A_m)^{-1}B$.
%By Assumption \ref{mo:as:1} $C^TB$ is full rank, and by the definition of  $M$ in \eqref{M} we have $M^TC^TB=B^TCM>0$. A similar explicit construction of an $L_s$ such that $M^TC^T(sI-A_m-L_sC^T)^{-1}B$ is SPR can be found in \cite{huang}.
%\end{rem}

\begin{lem}\label{LEM4}
Choosing $L=L_s-\rho BM^T$ where $L_s$ is defined in Lemma \ref{LEM2} and $\rho>0$ is arbitrary, the transfer function $M^TC^T(sI-A_m-LC^T)^{-1}B$ is SPR and satisfies:
\be
\begin{split}
\label{eq:ly4p}
(A_m+LC^T)^TP+P(A_m+LC^T)  = -Q \\
 Q   \triangleq Q_s  + 2\rho  C MM^TC^T \\ 
 PB=CM
\end{split}
\ee
where $P=P^T>0$ and $Q_s=Q_s^T>0$ are independent of $\rho$ and $M$ is defined in \eqref{M}.
\end{lem}

%
%\label{lem: M}For a square system in the form of (\ref{eq:Wsu}),
%there exists a $M\in\mathbb{\Re}^{p\times p}$ such that the post-compensated
%system
%\begin{equation}
%W_{spd}(s)=M^{T}C^{T}(sI-A_{m})^{-1}B\label{eq:Wspd}
%\end{equation}
%is minimum phase and satisfies
%\begin{equation}
%M^{T}C^{T}B=(M^{T}C^{T}B)^{T}>0.\label{eq:SPM}
%\end{equation}
%\end{lem}
%\begin{IEEEproof}
%It has been proved in \cite{sab90} that connecting a invertible post-compensator
%will not alter system's zero locations, i.e. $zero(C^{T}(sI-A_{m})^{-1}B)=zero(M^{T}C^{T}(sI-A_{m})^{-1}B)$.
%Therefore the system $M^{T}C^{T}(sI-A)^{-1}B$ is minimum phase. Since
%the zeros hold for any invertible post-compensator, one can simply
%pick a $M^{T}=(C^{T}B)^{-1}$ and it satisfies the additional requirement
%$M^{T}C^{T}B=(M^{T}C^{T}B)^{T}>0$.
%\end{IEEEproof}

We should note that the $B$ matrix above corresponds to additional $p-m$ inputs which are fictitious. The following corollary helps in determining controllers that are implementable.

%\begin{cor}
%\label{cor:Ls}
%For a square system in the form of (\ref{eq:Wspd})
%that is minimum phase and satisfies (\ref{eq:SPM}), there exists
%a $L_{s}$ such that the system $M^{T}C^{T}(sI-A_{m}-L_{s}C^{T})^{-1}B$
%is SPR.\end{cor}
%\begin{IEEEproof}
%The proof is presented in \cite{yu10}\cite{RefWorks:90} and therefore
%is omitted here. \end{IEEEproof}
%\begin{rem}
%Following Lemma \ref{LEM3}, we can pick a $L=L_{s}-\rho BM^{T}$
%such that the system $W_{vcl}(s)=M^{T}C^{T}(sI-A_{m}-LC^{T})^{-1}B$
%is still SPR. The two equations of KYP Lemma, as in (\ref{eq:ly3p})
%and (\ref{eq:ly2p}), also holds for $W_{vcl}(s)$. 
%\end{rem}
%We have shown with proper feedback, the virtual system (\ref{eq:Wspd})
%can be made SPR. The following theorem will derive the SPRness of
%the real system from the SPRness of the virtual system.
%\begin{thm}

\begin{cor}\label{corl}
Choosing $L=L_s-\rho BM^T$ where $L_s$ is defined in Lemma \ref{LEM2s} and $\rho>0$ is arbitrary, the transfer function $M_1^TC^T(sI-A_m-LC^T)^{-1}B_1$ is SPR and $M_1$ is defined by the partition $M=[M_1\  M_2]$ which satisfies $P[B_1\ B_2]=C[M_1
\ M_2]$.\end{cor}

%
%\label{thm:Square down}For a non-square system in the form of (\ref{eq:non-sqaure system})
%and (\ref{eq:ref system}), if there exists a $L\in\mathbb{R}^{n\times p}$
%and a $M\in\mathbb{R}^{p\times p}$ and a $B_{2}\in\mathbb{R}^{n\times(p-m)}$
%such that the virtual system $W_{vcl}(s)=M^{T}C^{T}(sI-A_{m}-LC^{T})^{-1}B$
%is strictly positive real (SPR), then there exists a $M_{1}\in\mathbb{R}^{p\times m}$
%such that the real system $W_{cl}(s)=M_{1}^{T}C^{T}(sI-A_{m}-LC^{T})^{-1}B_{1}$
%is SPR.\end{thm}
%\begin{IEEEproof}
%The details proof has been presented in \cite{RefWorks:90}. Here
%we only list important intermediate results. Since $B=[B_{1},\ B_{2}]$,
%$M$ can also be parted as $M=[M_{1},\ M_{2}]$. The following equality
%holds $P[B_{1},\ B_{2}]=C[M_{1},\ M_{2}]$. It follows that 
%\begin{equation}
%PB_{1}=CM_{1}\label{eq:PB1}
%\end{equation}
%Together with (\ref{eq:ly3p}), the two equations of KYP lemma hold
%and $W_{cl}(s)$ is indeed SPR.
%\end{IEEEproof}
Accordingly, we propose the following adaptive law:
\begin{equation}
\begin{array}{c}
\dot{\Theta}=-\Gamma_{\theta}x_{m}e_{y}^{T}M_{1}\\
\dot{K}=-\Gamma_{k}re_{y}^{T}M_{1}
\end{array}\label{eq:Ada Law}
\end{equation}
The following theorem shows that the overall system is globally stable and $\lim_{t\to\infty} e(t)=0$.
\begin{thm}
The closed-loop adaptive system specified by \eqref{eq:non-sqaure system},
\eqref{eq:ref system}, \eqref{eq:cin} and \eqref{eq:Ada Law}, satisfying
assumptions 2 to 7, with $B$ chosen as in \eqref{Bdef}, $L$ as in Lemma \ref{LEM4}, $M$
chosen as in Equation \eqref{M}, with $M_{1}$ defined in Corollary \ref{corl}, and $\rho>\rho^{*}$ has globally bounded
solutions with $\lim_{t\to\infty}e_y(t)=0$, where $\rho^{*}$ is defined as
\begin{equation}
\rho^{*}=\frac{\bar{\lambda}^{2}\bar{\theta}^{*2}\left\Vert M_{1}\right\Vert ^{2}}{2\lambda_{min}(Q_{s})\lambda_{min}(MM^T)}.
\end{equation}
\end{thm}
\begin{IEEEproof}
The proof follows as in that of Theorem \ref{thm1}.
%
% by considering the error dynamics
%\begin{equation*}
%\begin{array}{c}
%\dot{e}=(A_{m}+LC^{T})e+B_{1}\Lambda(\tilde{\Theta}^{T}x_{m}+\tilde{K}^{T}r-\Theta^{*T}e)\\
%e_{y}=C^{T}e
%\end{array}
%\end{equation*}
%which leads to the Lyapunov function in \eqref{disc:v00} with
%\begin{equation*}
%\dot{V}= - \mathcal E^T \mathcal Q(\nu) \mathcal E\label{eq:Vdot}
%\end{equation*}
%where 
%\[
%\mathcal Q(\rho)=\left[\begin{array}{cc}
%2\rho MM^{T} & M_{1}\Lambda\Theta^{*T}\\
%\Theta^{*}\Lambda M_{1}^{T} & Q_{s}
%\end{array}\right]\qquad\mathcal E=\left[\begin{array}{c}
%e_{y}\\
%e
%\end{array}\right]
%\]
%Using Schur complement, it is easy to show that $\dot{V}\leq0$ if
%$\rho>\rho^{*}$.
\end{IEEEproof}

%
%\subsection{Other Extensions without Proof}
%
%We now discuss some trivial extensions.
%
%
%
%The one extension still eluding the authors is how this method can be extended when $C^TB=0$. We believe the most direct path forward would be to 

\section{Simulation Study}\label{sec:example}

For the simulation study we compare the performance of a combined linear and adaptive LQG controller to an LQR controller, which is full states accessible by definition.  The uncertain system to be controlled is defined as  
\ben
\dot x_p = A_p x_p + B_p u \quad \text{and} \quad y_p=C_y^T x_p
\een
where $x_p=\bb V & \alpha & q & \theta \eb^T$ is the state vector for the plant consisting of: velocity in ft/s, angle of attack in radians, pitch rate in radians per second, and pitch angle in radians. The control input consists of $u= \bb T & \delta \eb^T$, the throttle position percentage and elevator position in degrees. The measured outputs are   $y_p= \bb V & q & h\eb^T$ where $h$ is height measured in feet. We note that two of the states for this example are not available for measurement, the angle of attack and the pitch angle. The pitch angle is never directly measurable and is always reconstructed from the pitch rate through some filtering process. The angle of attack however is usually available for direct measurement in most classes of aircraft. There are several classes of vehicles however where this information is hard to obtain directly: weapons, munitions, small aircraft, hypersonic vehicles, and very flexible aircraft, just to name a few.

 In this example we intend to control the altitude of the aircraft, and for this reason an integral error is augmented to the plant. The extended state plant is thus defined as
\ben
\dot x = A x +B_1u+ B_z r \quad \text{and} \quad  y = C^Tx
\een
where $y_z=h$, $r$ is the desired altitude, 
\ben\begin{split}
x&=\bb x_p \\ \int(y-r)  \eb, \quad A = \bb A_p & 0_{4\times 1} \\ C_z & 0_{1\times 1}  \eb, \quad B_1=\bb B_p \\ 0_{1\times 2} \eb, \\ B_z &= \bb 0_{4\times 1} \\ -I_{1\times 1} \eb,\quad  C^T=\bb C^T & 0_{3\times 1} \\ 0_{1\times 4} & I _{1 \times 1}\eb, \quad y= \bb  y_p \\ \int (y_z-r) \eb
\end{split}
\een

The reference system is defined as
\ben
\dot x_m = A_m x_m  + B_z r - L_\nu (y-y_m) \quad \text{and} \quad  y_m = C^Tx_m
\een
where  $A_m = A_{nom}+ B_1 K^T_{R}$, with $K_R^T=-R_R^{-1}B_pP_R$  the  solution to the algebraic Riccati equation
\ben
A_{nom}^TP_R + P_R A_{nom} -P_R B R_R^{-1} B^T P_R + Q_R=0
\een
and 
\ben
A_{nom} = \bb A_{p,nom} & 0_{4\times 1} \\ C_z & 0_{1 \times 1}  
\eb.\een 
The closed-loop reference model gain $L_\nu$ is defined as in \eqref{mod:l22} where we have squared up the input matrix  through the artificial selection of a matrix $B_2$ and defined $B=[B_1 \ B_2]$ so that $C^TB$ is square, full rank, and $C^T(sI-A_m)^{-1}B$ is minimum phase.
The control input for the linear and adaptive LQG controller is defined as
\ben
u=K_R^T x_m + \Theta^T x_m
\een
where the update law for the adaptive parameters is defined as 
\ben
\dot \Theta = - \Gamma x_m e_{y}^T M_1,
\een
with $M_1$ the first $m$ colums of $R_0^{-1/2}W$ where $W$ is defined just below \eqref{eq:W} . The LQR controller is defined as 
\ben
u=K_R^T x.
\een
All simulation and design parameters are given in Appendix \ref{app:example}. Note that the free design parameter $\Gamma$ has zero for the last entry, this is due to the fact that for an uncertainty in $A_p$ feedback from the integral error state is not needed for a matching condition to exist. The simulation results are now presented.

Figure 1 contains the trajectories of the state space for the adaptive controller (black), linear controller (gray), reference model  $x_m$ (black dotted), and reference command height (gray dashed). The reference command in height was chosen to be a filtered step, as can be seen by the gray dashed line. The plant when controlled only by the full state linear optimal controller is unable to maintain stability as can be seen by the diverging trajectories. The reference model trajectories are only visibly different from the plant state trajectories under adaptive control in the angle of attack subplot and the pitch angle subplot, the two states which are not measurable. Figure 2 contains the control input trajectories for the adaptive controller and Figure 3 contains the adaptive control parameters. There are two points to take away form the simulation example. First, the adaptive output feedback controller is able to stabilize the system while the full state accessible linear controller is not. Second, the state trajectories, control input, and adaptive parameters exhibit smooth trajectories. This smooth behavior is rigorously justified in \cite{gib13access} for a simpler class of closed-loop reference models.

\begin{figure}[h!]
\centering
\includegraphics[width=3.3in]{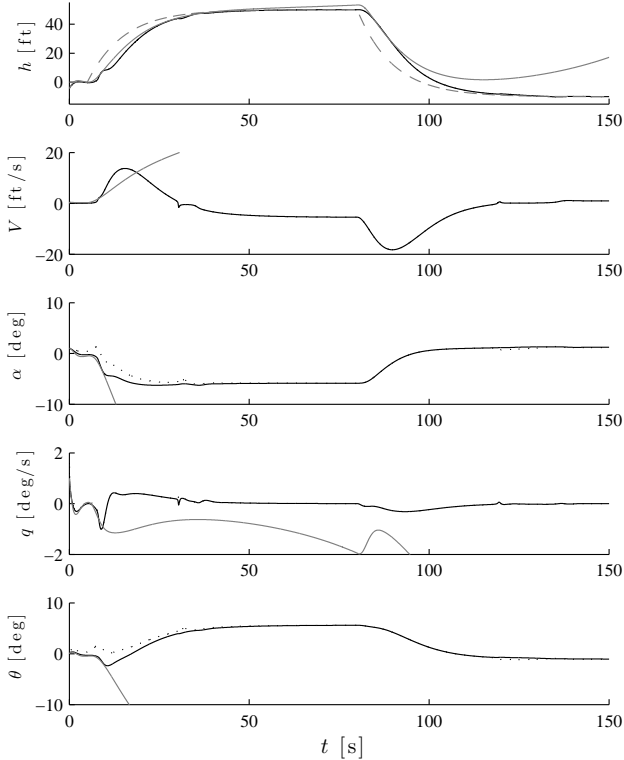}
\caption{Trajectories in state space from the adaptive controller (black), linear LQR controller (gray), reference model  $x_m$ (black dotted), reference command for height (gray dashed).}\label{fig:1}
\end{figure}

\begin{figure}[h!]
\centering
\includegraphics[width=3.3in]{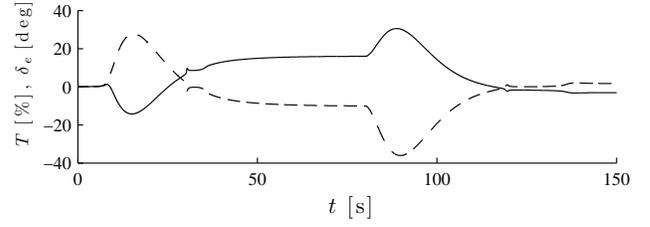}
\caption{Control inputs from the adaptive controller, throttle percentage (dashed) and elevator position (solid).}\label{fig:2}
\end{figure}

\begin{figure}[h!]
\centering
\includegraphics[width=3.3in]{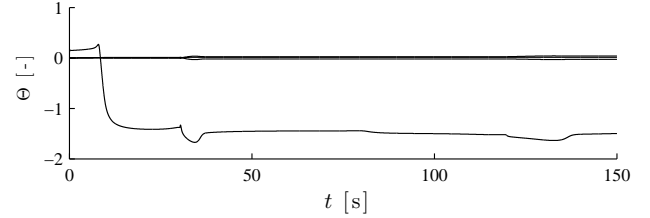}
\caption{Adaptive Parameters.}\label{fig:3}
\end{figure}

\section{Conclusions}
This note presents methods for designing output feedback adaptive controllers for plants that satisfy a states accessible matching condition, thus recovering a separation like principle for this class of adaptive systems, similar to linear plants.
\bibliographystyle{IEEEtran}

\appendices
\section{The SPR condition, KYP Lemma and Transmission Zeros}\label{app:lin}
%\begin{defn}[\cite{tac:ioa87}]
%A rational function $H(s)$ is Strictly Positive Real (SPR) iff
%\begin{itemize}
%\item $H(s)$ is analytic in $\text{Re}[s]\geq 0$,
%\item $\text{Re}[H(j\omega)]>0$ $\forall \ \omega\in(-\infty,\infty)$ and,
%\item $\lim_{\omega^2\to\infty} \omega^2 \text{Re}[H(j\omega)]>0$ when the relative degree is 1.
%\end{itemize}
%When $n^*=0$ the third condition is not needed.
%\end{defn}
%\begin{lem}[Meyers Kalman Yakubovich (MKY)]\label{lem:mky}
%Given  vectors $b$ and $h$, and asymptotically stable matrix $A$, and symmetric positive-definite matrix $L$, if 
%\ben
%\text{Re}[H(i\omega)]= \text{Re} \left [  h^T(i \omega I - A)^{-1} b \right ]>0\quad  \forall \  \omega\in(-\infty,\infty),
%\een
%then there exists a scalar $\epsilon>0$, a vector $q$ and $P=P^T>0$ such that 
%\begin{align*}
%A^TP+PA&=-qq^T-\epsilon L\\
%Pb-h&=0.
%\end{align*}
%\end{lem}
%\begin{proof} See \cite[Lemma 2.4]{annbook}
%\end{proof}

This section contains relevant definitions for linear systems that were assumed to be familiar to the reader. They have been included for completeness. We begin with two definitions of positive realness. The KYP Lemma is then introduced. The section closes with a few rank conditions related to transfer matrices.

\begin{defn}[{\cite{akylemma,annbook}}]
An $n\times n$ matrix $Z(s)$ of complex variable $s$ is {\em Positive Real} if
\begin{enumerate}
\item $Z(s)$ is analytic when $\text{Re}(s)>0$ (Re $\triangleq$ real part)
\item $Z^*(s)=Z(s^*)$ when $\text{Re}(s)>0$ ($^*$ denotes complex conjugation)
\item $Z^T(s^*)+Z(s)$ is positive semidefinite for $\text{Re}(s)>0$. 
\end{enumerate}
\end{defn}
\begin{defn}
An $n\times n$ matrix $Z(s)$ of complex variable $s$ is {\em Strictly Positive Real} (SPR) if $Z(s-\epsilon)$ is positive real for some $\epsilon>0$
\end{defn}
Throughout the remainder of this section the following transfer matrix is referred to
\be\label{app:lin1}
Z(s)  = C^T(sI-A)^{-1}B.
\ee

\begin{lem}[Kalman Yakubovich Popov (KYP), {\cite[Lemma 2.5]{annbook}}]\label{lem:aky}
A $Z(s)$ as defined in \eqref{app:lin1} that is minimal is SPR iff there exists $P=P^T>0$ and $Q=Q^T>0$ s.t.
$A^TP+PA=-Q$ and $PB=C$.
\end{lem}

\begin{cor}\label{cor:spr}
If $B\in \Re^{n \times m},\ m\leq n$ is rank $m$ and $Z(s)$ is SPR, then $C^TB = (C^TB)^T >0$. \end{cor} 
\begin{proof} Given that $PB=C$, it also follows that $B^TP=C^T$ and thus $B^TPB=C^TB$ is symmetric, rank $m$ and positive.\end{proof}

%\begin{defn} The matrix pencil of the triple $\{A,B,C^T\}$ for the system $Z(s) = C^T(sI-A)^{-1}B$ is defined as 
%\be\label{eq:pencil}
%P(s) = \bb sI-A & B \\ -C^T & 0\eb
%\ee
%\end{defn}
%
%\begin{defn}
%The transmission zeros for a minimal $Z(s) = C^T(sI-A)^{-1}B$  are defined as the set $\mathcal Z_t = \{s | \text{rank}\ P(s) < n+ \text{min}(m,p)\}$
%where $P(s)$ is the matrix pencil  defined in \eqref{eq:pencil}.\cite{dav74}
%\end{defn}

\begin{defn}
For $Z(s)$ as defined in \eqref{app:lin1} that is minimal and square, the {\em transmission zeros} are the zeros of the polynomial $\psi (s)   =  \det (sI- A) \det [C^T(sI-A)^{-1}B]$ \cite[Theorem 1.19]{kwabook}.
\end{defn}

\begin{lem}\label{lem:inv}
For $G\in \Re^{m\times m}$ and full rank, the location of the transmission zeros for a square $Z(s)$ in \eqref{app:lin1} are equivalent to the location of the transmission zeros of $G Z(s)$.
\end{lem}
\begin{proof}
If ${s_0\in \Ce}$ is a transmission zero, then
$\det (s_0I- A) \det [GC^T(s_0 I-A)^{-1}B] =0$,
and recalling the product rule for determinates $ \det [GC^T(s_0I-A)^{-1}B] =  \det (G) \det [C^T(s_0I-A)^{-1}B]$. $G$ is full rank and thus $\det (G)\neq 0$. Therefore, $s_0$ is a solution to $\det (s_0I- A)  \det [C^T(s_0 I-A)^{-1}B]=0$ as well. \end{proof}

\section{Parameters for Section \ref{sec:example}}\label{app:example}
The plant parameters are given as:
\ben
\begin{split}
A_{p,nom} &= \bb -0.038 &18.94& 0& -32.174 \\
   -0.001 & -0.632 & 1 & 0 \\
   0 & -0.759 & -0.518 & 0 \\
   0 & 0 & 1 &  0 \eb \\
B_p &=\bb10.1 &0\\ 
    0& -0.0086\\
    0.025& -0.011\\ 
    0 &0\eb \\
C_y & = \bb1& 0& 0& 0 \\
    0 &0 &1 &0 \\
    0 &-250& 0& 250\eb \\
C_z & = \bb 0 &-250 & 0 & 250 \eb \\
  A_p&= A_{p,nom} + B_p\bb -2 &1.5& 2& -2 \\ 1.5 &-2& 2& 1 \eb  \end{split}
\een
The linear control design parameters:
\ben\begin{split}
Q_R&= \text{diag}(\bb1& 1& .1& 0 & .1\eb)\\
R_R&=  \text{diag}(\bb1& 10\eb)\\
\end{split}\een
where $K_R^T=-R_R^{-1}B_pP_R$ with $P_R$ the solution to the control Riccati equation.

The adaptive control design
\ben\begin{split}
Q_0&=I_{(n+q) \times (n+q) }\\
R_0&=I_{(p+q) \times (p+q) }\\
\Gamma&= \text{diag}(\bb 1& 1& 1& 1& 0\eb)\\
\nu&=0.01\\
B_2& = \bb0 &0\\
      0 &1\\
      3 &0\\
      0 &3\\
      1 &0\eb.
\end{split}\een

\end{document}